\documentclass[a4paper,11pt]{article}
\usepackage{graphicx, amssymb, amsmath}
\usepackage{amsthm}
\usepackage{hyperref}
\usepackage{epsfig, color, graphicx}
\usepackage{enumerate}
\usepackage{algorithm2e}
\usepackage[margin=1.0in]{geometry}

\newtheorem{theorem}{Theorem}
\newtheorem{cor}{Corollary}
\newtheorem{lemma}{Lemma}
\newtheorem{nc}{Necessary Condition}
\newtheorem{nc'}{Necessary Condition}
\title{Partitions of planar point sets into polygons}
\author{Ajit Arvind Diwan \thanks{Department of Computer Science and Engineering,
        Indian Institute of Technology Bombay, {\tt aad@cse.iitb.ac.in}}
        \and
        Bodhayan Roy \thanks{Department of Computer Science and Engineering, Indian Institute of Technology Bombay, {\tt  broy@cse.iitb.ac.in}}}

\begin{document}
\thispagestyle{empty}
\maketitle
\begin{abstract}
In this paper, we characterize planar point sets that can be partitioned into disjoint polygons of arbitrarily specified sizes.
We provide an algorithm to construct such a partition, if it exists,
in polynomial time.
We show that this problem is equivalent to finding a specified $2$-factor in the visibility graph of the point set.
The characterization for the case where all cycles have length $3$ also translates to finding 
a $K_3$-factor of the visibility graph of the point set.
We show that the generalized problem of finding a $K_k$-factor
of the visibility graph 
of a given point set for $k \geq 5$ is NP-hard.
\end{abstract}
\section{Introduction}
Partitioning of point sets is a well studied topic in Computational Geometry.
Let $P$ be a set of finite number of points in the plane. 
A partition of $P$ is called a \emph{convex partition} if $P$ is partitioned into $j$ subsets
$S_1, S_2, \ldots, S_j$ such that all the points of $S_i$ form the vertices 
of a convex polygon \cite{conv-dec}. Problems concerning such partitions
have been studied, and 
bounds on the number of sets required for such a disjoint partition have been established \cite{conv-dec, urabe-1996, aichh-2007}.
In this paper, 
we study the following two related partitions of $P$, and the equivalent problem on the visibility graph of $P$. 

$\\ \\$
A partition of $P$ into subsets $S_1, S_2, \ldots, S_j$ is said to be a \emph{cycle partition} of $P$,
when the points of each $S_i$ can be joined by straight line segments to form a simple polygon, i.e. 
no $S_i$ has all points collinear.
We say that two points $p_i$ and $p_j$ of $P$ 
are \emph{visible} to each other if the line 
segment $p_ip_j$ does not contain any other point of $P$. In other words,
$p_i$  and $p_j$ are visible to each other if $P \cap  {p_ip_j}=\{p_i, p_j\}$. If two points are not visible, 
they are called  \emph{invisible} to each other.
If a point $p_k \in P$ lies on the segment $p_ip_j$ connecting two points $p_i$ and $p_j$ in $P$, 
we say that $p_k$ blocks the visibility between $p_i$ and $p_j$, and
$p_k$ is called a \emph{blocker} in $P$.
A partition of $P$ into subsets $S_1, S_2, \ldots, S_j$ is said to be a \emph{clique partition} of $P$,
when all the points of each $S_i$ are mutually visible, i.e. 
no $S_i$ has three collinear points, and no two points of $S_i$ are blocked by points from any other $S_l$.

$\\ \\$
The \emph{visibility graph} (also called the \emph{point visibility graph} (denoted as PVG) of $P$ is defined by associating a vertex $v_i$ with each point $p_i$ of $P$ 
and such that $(v_i, v_j)$ is an undirected edge of the PVG if  $p_i$ and $p_j$ are visible to each other. 
Point visibility graphs have been studied in the contexts of 
construction \cite{cgl-pgd-85, Edelsbrunner:1986:CAL}, 
recognition \cite{prob-ghosh, recogpvg-2014, pvg-card, pvg-np-hard}, 
connectivity \cite{viscon-wood-2012},
chromatic number and clique number \cite{kpw-ocnv-2005, p-vgps-2008}. 
 Let $H$ be a connected graph. For a given graph $G$, an $H$-factor of $G$ is a spanning subgraph of $G$ whose components are isomorphic to $H$. 
 Thus, a $C_k$-factor of $G$ is a spanning subgraph of $G$ whose components are isomorphic to the cycle on $k$ vertices.
 Similarly, a \emph{$K_k$-factor} of $G$ is a spanning subgraph of $G$ whose components are isomorphic to the clique on $k$ vertices.
 To decide whether or not a graph $G$ on $kn$ vertices has a $K_k$-factor, is NP-hard for $k \geq 3$ \cite{cai-79}. 

 $\\ \\$
We say that a cycle partition of a point set is \emph{disjoint} when no two of the polygons enclosed by the cycles intersect with respect to 
vertices, edges or area.
In this paper, in Section \ref{secnoncr}, we study disjoint cycle partitions of point sets.
In Section \ref{subsectrifac} we study the special case where all cycles are of length $3$.
We provide a necessary and sufficient condition for this case.
We also show that all point sets that admit any cycle partition into cycles of length $3$, also admit such a disjoint cycle partition.
In Section \ref{subseccyfac} we study the generalized disjoint cycle partitions except the case mentioned above,
We provide a different necessary and sufficient condition for it, thereby completely characterizing all point sets
that admit a disjoint cycle partition.
We also show that this problem is equivalent to finding a specified $2$-factor in the visibility graph of the point set,
and characterize all such PVGs.
In Section \ref{secclfac} we study the problem of clique partitions of point sets.
If all cliques are of size $3$, then this problem is the same as the special case in cycle partition.
We prove that if all cliques are of size $5$ or more, then the problem becomes NP-hard.
Finally, in Section \ref{secconcrem}, we conclude with some remarks and open questions.
\section{Disjoint cycle partitions} \label{secnoncr}
 Let $P$ be a given set of finitely many points in the plane. 
 Denote the 
convex hull of a point set $P$ by $CH(P)$.
 We have the following lemmas on $P$.
 \begin{lemma} \label{lem1}
  Every point set $P$ with even number of points has a non-crossing matching.
 \end{lemma}
\begin{proof}
 Remove a vertex $p_i$ of $CH(P)$. 
 Remove a vertex $p_j$ of $CH(P \setminus \{p_i\})$ adjacent to $p_i$ and match the two points.
 The two matched points $p_i$ and $p_j$ clearly lie outside of $CH(P \setminus \{p_i \cup p_j\})$.
   Repeat this process to get the perfect matching.
\end{proof}
$\\ \\$
Now we generalize the concept of a non-crossing matching.
A partition of $P$ into subsets $S_1, S_2, \ldots, S_j$ is said to be a \emph{cycle partition} of $P$,
when the points of each $S_i$ can be joined by straight line segments to form a simple polygon, i.e. 
no $S_i$ has all points collinear.
\subsection{A necessary condition for disjoint triangle partition}
A cycle partition of $P$ is said to be a \emph{triangle partition}
when all the cycles have length $3$.
Consider a subset $S$ of $P$ such that no points of $S$ are visible from each other.
Then we call $S$ an \emph{independent set} of $P$.
Observe that $S$ also induces an independent set in the visibility graph of $P$.
We have the following necessary condition for point sets that admit a triangle partition.
\begin{nc} \label{nc1}
Let $P$ be a set of $3n$ points. If $P$ admits a triangle partition then $P$ does not have an independent set of $n+1$ points.
\end{nc}
\begin{proof}
 Assume on the contrary that $P$ admits a triangle partition and also has an independent set $S$ of $n+1$ points.
 The triangle partition of $P$ must have $n$ triangles. 
 So, some two points of $S$ must lie in the same triangle, which is impossible.
\end{proof}
\subsection{Sufficiency of the necessary condition} \label{subsectrifac}
Here we provide a characterization of all 
point sets that admit a triangle partition,
 by proving the sufficiency of Necessary Condition \ref{nc1}.
\begin{figure}[h] 
\begin{center} 
\centerline{\hbox{\psfig{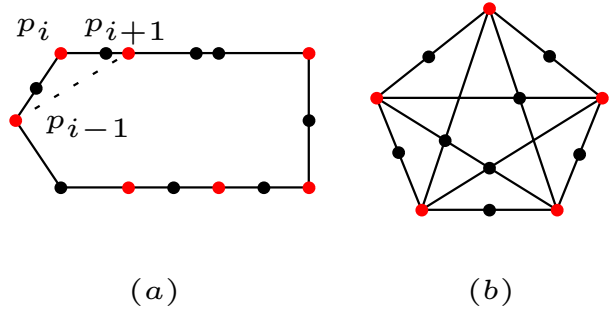}}}
\caption{(a) A point set with the points of $\mathcal{I}$ on three different lines on the convex hull. The points of 
$\mathcal{I}$ are coloured red. (b) A pentagon with five vertices from $\mathcal{I}$.}
\label{pentagonfig}
\end{center}
\end{figure} 
\begin{lemma}
The point set $P$
  on $3n$ points can have an independent set of size $n+1$ only if one of the following is true.
  \begin{enumerate}
   \item All the points of the independent set are collinear.
   \item All the points of the independent set lie on two lines on $CH(P)$.
  \end{enumerate}
  \end{lemma}
  \begin{proof}
Suppose not all the points of the independent set (say, $\mathcal{I}$) of size $n+1$ in $P$ are collinear.
Consider a maximal plane graph $G_I$ with points of $\mathcal{I}$ as its vertices.
%
Since each edge between two vertices of $G_I$ is 
basically a straight line segment between corresponding points in $P$, it must contain
a blocker in $P$ but not in $\mathcal{I}$.
There can be at most $2n-1$ such blockers.
Let $\vert CH(\mathcal{I}) \vert = h$.
Since $\vert \mathcal{I} \vert = n+1$, we get the following relation.
 \begin{eqnarray*}
  & &  2n-1 \geq 3 \vert \mathcal{I} \vert - h - 3 \\
  & \Rightarrow &  2n-1 \geq 3n - h  \\
   & \Rightarrow & h \geq n + 1
\end{eqnarray*}
Hence, all the points of $\mathcal{I}$ must lie on $CH(\mathcal{I})$.
Observe that it also follows from above that all the vertices of $CH(P)$ must also be in $\mathcal{I}$.
Now suppose $CH(P)$ is (i) a polygon with at least five vertices, or (ii) with points of $\mathcal{I}$ lying in 
the interior of at least three of its edges (Figure \ref{pentagonfig}(a)).
This means, either $CH(P)$ has property (i), 
or we can always exclude a vertex $p_i$ from $CH(P)$ such that  $p_i \in \mathcal{I}$ and $CH(P \setminus \{ p_i \})$ has property (i) or (ii).
Let $p_i$ be a vertex of $CH(P)$. Also, let $p_{i-1}$, $p_i$ and $p_{i+1}$ be consecutive points on $\mathcal{I}$ on $CH(P)$.
Exclusion of $p_i$ results in the exclusion of the blockers 
on $p_{i-1}p_i$ and $p_ip_{i+1}$. A blocker on $p_{i-1}p_{i+1}$ is now a new convex hull point, and there are at most $n-3$ blockers in the interior of the
new convex hull, satisfying all blocking relationships. We continue the process till we are left with the remaining points of $\mathcal{I}$ forming 
a pentagonal convex hull (Figure \ref{pentagonfig}(b)). 

$\\ \\$
Now, for the pentagon, we must block five pairs of convex hull vertices with only two blockers. This requires the line segments passing through some three of these 
pairs to coincide at one point. But this is impossible, because there are only five convex hull points constituting the five pairs. Hence, we have a contradiction.
  \end{proof}
\begin{figure}[h] 
\begin{center} 
\centerline{\hbox{\psfig{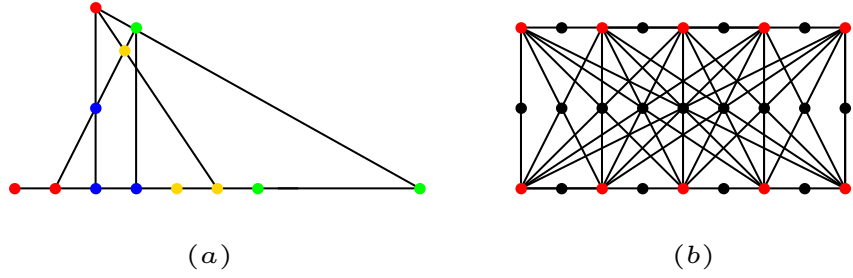}}}
\caption{(a) A point set where $\mathcal{I}$ has $n$ collinear points. (b) A point set where $\mathcal{I}$ has $n+1$ points on two different lines.}
\label{examplefig}
\end{center}
\end{figure}
\begin{theorem}
\label{trndec}
 If a point set $P$ on $3n$ points does not have an independent set of $n+1$ points, then $P$ has a triangle partition.
\end{theorem}
\begin{proof}
Remove a point $p_i$ from the convex hull of $P$. If all other points of $P$ are collinear, then it violates the necessary condition.
So $P \setminus \{ p_i\}$ 
must have a polygonal convex hull. From $CH(P \setminus \{ p_i\})$, remove two more consecutive points facing $p_i$,
such that at least one of them is on the boundary of $CH(P)$. Call these two points $p_j$ and $p_k$. Suppose that 
after the removal of these three points, $P$ on $3(n-1)$ points violates the necessary condition.

$\\ \\$
 First consider the case where $n$ points of $P\setminus \{ p_i, p_j, p_k\}$ form an independent set and are collinear on some line $l$. 
 Since $P$ does not violate 
 the necessary condition, $l$ must have either $2n-1$ or $2n$ points. 
 Divide $l$ into segments $\{s_1, s_2, \ldots, s_n \}$ from left to right, of two consecutive points each, possibly except $s_n$. 
 Starting from the leftmost segment $s_1$, consider the free
 point $p_i$ above $l$ that makes the maximum angle with the rightmost point of $s_1$. Make a triangle with $p_i$ 
 and the two points of $s_1$. Repeat the process for all segments. If points above $l$ are depleted, then 
 use the points below $l$ in an analogous process to complete the partition (Figure \ref{examplefig}(a)). 
 It may so happen that the last segment $s_n$
 has only one point. In this case, use it to form a triangle with two remaining points $p_j$ and $p_k$ outside of $l$. Suppose that 
 the remaining two points are on different sides of $l$. In this case, use the point previously allotted to  
 $s_{n-1}$, and a free point on the same side of $l$, (say, $p_j$) to form the last triangle. 
 Now form the remaining triangle with $p_k$ and $s_{n-1}$.

 $\\ \\$
 Next consider the case where $n$ points of $P\setminus \{ p_i, p_j, p_k\}$ form an independent set $\mathcal{I}$ and are on two 
 lines (say, $l_1$ and $l_2$) on the boundary of $CH(P \setminus \{ p_i, p_j, p_k\})$, 
 violating the necessary condition for $P\setminus \{ p_i, p_j, p_k\}$.
  Suppose that three of $\{ p_i, p_j, p_k\}$ are on $l_1$ or $l_2$. Since $P$ does not violate the necessary condition, 
  none of $\{ p_i, p_j, p_k\}$ can be added to $\mathcal{I}$. However, one of these three points 
  can form another independent set of size $n$ with 
  the blockers between the points of $\mathcal{I}$ on $l_1$ and $l_2$. Observe that $l_1$ and $l_2$ are on the boundary
  of $P$. Wlog let $\{ p_i, p_j, p_k\}$ lie on the right end of $l_1$ and $l_2$. We remove three different points (say, $\{ p_a, p_b, p_c\}$)
  from the left end of $l_1$ and $l_2$, among which one point must be in $\mathcal{I}$. So, $P\setminus \{ p_a, p_b, p_c\}$
  satisfies the necessary condition. Similarly, even if not all three among $\{ p_i, p_j, p_k\}$ are on $l_1$ or $l_2$,
  one end of $l_1$ and $l_2$ each are on the boundary of $P$. As before, remove a point (say, $p_a$) from this end of $l_1$ or $l_2$,
  which must also be on $CH(P)$. Compute $CH(P \setminus \{ p_a\})$ and remove two more points (say, $p_b$ and $p_c$) 
  from the chain of $CH(P \setminus \{ p_a\})$
  facing $p_a$. Since $p_b$ and $p_c$ must also come from $l_1$ and $l_2$, $P \setminus \{ p_a, p_b, p_c\}$ must satisfy the necessary condition.
\end{proof}
$\\ \\$
The above proof is constructive in nature and gives a polynomial algorithm to construct a triangle partition of $P$.
\begin{cor}
 A point visibility graph $G$ of order $3n$ has a $K_3$-factor
if and only if $G$ does not have an independent set
of size $n+1$.
\end{cor}
\subsection{Generalized cycle partitions} \label{subseccyfac}
Let $S$ be a given set of cycles $\{ C_1, C_2, \ldots, C_l \}$, not all $C_i$ of length $3$.
Let $L_i$ be the length of $C_i \in S$.
In this section, we show that it is possible to partition a point set $P$ into disjoint cycles of $S$ if and only if $P$ 
does not have $\sum_{i=1}^l  L_i - l +1$ collinear points.
\begin{figure}[h] 
\begin{center} 
\centerline{\hbox{\psfig{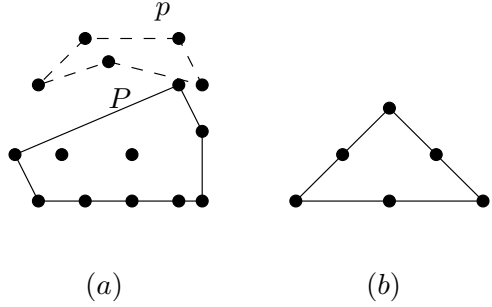}}}
\caption{(a) A $C_5$ separated out from $P$. (b) A set of $6$ points that cannot be triangle partitioned.}
\label{sepcyclfig}
\end{center}
\end{figure}
\begin{nc'}
 If $P$ on $\sum ^l _{i=1} L_i$ points can be partitioned into disjoint cycles of lengths $\{L_1, L_2, \ldots, L_l \}$,
 then at most $\sum ^l _{i=1} L_i- l$ points of $P$ are collinear.
\end{nc'}
\begin{proof}
 Each cycle $C_i$ can have at most $L_i-1$ collinear points. 
 So, if $P$ has more than $\sum ^l _{i=1} L_i$ collinear points, then some $C_i$  must contain all $L_i$ collinear points, which is impossible.
\end{proof}
\begin{lemma} \label{disjcycl}
If $P$ has $\sum_{i=1}^l  L_i$ points, not all collinear, then
 it is possible to separate out  
$C_i$ from $P$ so that $CH(C_i)$ and $CH(P \setminus C_i)$ are disjoint.
%
\end{lemma}
\begin{proof}
Choose any vertex $p$ of $CH(P)$ and let $C = \{ p\}$. Draw two tangents from $p$ to $CH(P \setminus \{ p \})$.   
Start from the left tangent on the chain of $CH(P \setminus \{ p \})$ facing $p$, traverse towards the right tangent.
Select consecutive points of $CH(P \setminus \{ p \})$ and add them to $C$ till $\vert C \vert = L_i$.
If $\vert C \vert < L_i$ even after all this points are added, then again consider a point that was
incident to one of the two tangents as the new $p$, and repeat 
the process till $\vert C \vert = L_i$.
Finally, obtain the edges of the cycle $C_i$ by ordering the points of $C$ in an angular order around $p$, joining each 
point by an edge to their predecessor and successor (Figure \ref{sepcyclfig} (a)).
\end{proof}
\begin{figure}[h]  
\begin{center}
\centerline{\hbox{\psfig{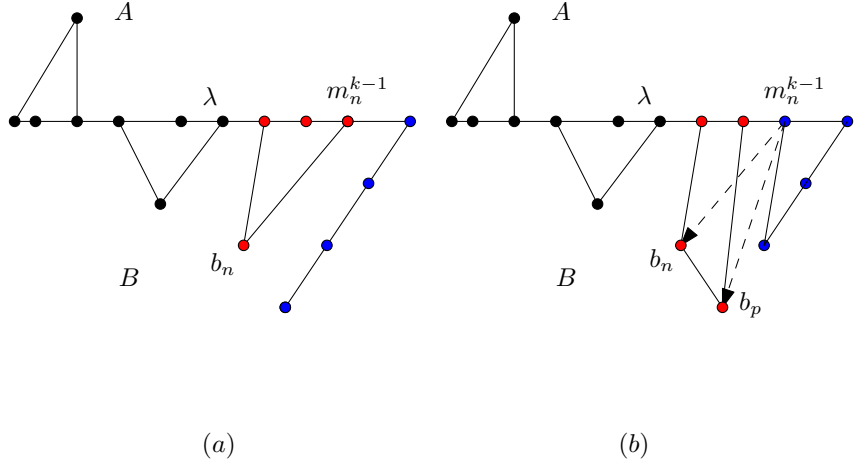}}}
\caption{(a) The free points of $B$ are collinear with the remaining point of $\lambda$. 
(b) One point of $\lambda$ is freed to be used in the $l'^{th}$ cycle.}
\label{linecyclfig}
\end{center}
\end{figure}
\begin{lemma} \label{bigline}
Let $S' = \{ C_1, C_2, \ldots, C_{l'} \}$, with some $C_i \neq C_{l'}$ of length more than $3$.
If $P'$ has $\sum _{i=1}^{l'} L_i$ points and there is some line $\lambda$ with at least 
 $\sum _{i=1}^{l'-1} L_i- l'+2$ and at most $\sum _{i=1}^{l'} L_i- l'$ collinear 
 points on it,
then $P$ can be partitioned into disjoint cycles of lengths specified by $l'$.
%
%
\end{lemma}
\begin{proof} 
 Let $\lambda$ be the line containing between $\sum _{i=1}^{l'-1} L_i- l'+2$ and $\sum _{i=1}^{l'} L_i- l'$
 points of $P$. Let the number of points on $\lambda$ be $x$.
  Then there must be $\sum _{i=1}^{l'} L_i -x$ points of $P$ outside of $l$. 
  Starting from the left, call the set of the first $L_1-1$ points as $M_1$, the next $L_2-1$ points as $M_2$ and so on, calling 
  the set of the last remaining points as $M_{l'}$ , where $1 \geq \vert M_{l'} \vert \geq L_{l'}-1$.
  For each $M_i$, let $m_i^j$ be the $j^{th}$ point of $M_i$ starting from the left.
  Call the set of all points over $\lambda$ as $A$, and the set of all points below $\lambda$ as $B$.
 
 $\\ \\$
  Suppose wlog that $\vert A \vert \leq l'-1$. 
Call the points of $A$ not assigned to any $M_i$ as \emph{free} points. Initially all points of $A$ are free.
Consider the set $A_1$ of all free points $a_u$ such that there is no point of $A$ in the interior of $a_um_1^1m_1^{L_1-1}$.
Consider the point $a_v \in A_1$ such that all other points of $A_1$ lie to the right side of $\overrightarrow{m_1^{L_1-1}a_v}$
and assign $a_v$ to $A_1$. Repeat this process with each $M_i$ and corresponding $A_i$ till none of the points of $A$ are free.
 
 $\\ \\$
Observe that at the end of the $i^{th}$ iteration of the above procedure, all the free points are 
to the right of $\overrightarrow{m_j^{L_j-1}}a_j$,
where $j \leq i$ and $a_j$ is the point of $A$ assigned to $M_j$. This can be seen by induction. The base case is clearly true,
because if $S' = \{ C_1\}$, then $L_1 \geq 4$.
If the claim is true for $i_1$ and not $i$, it means that a point $a_u$ lies to the left of $\overrightarrow{m_i^{L_i-1}}a_i$.
The point $a_u$ cannot lie inside $a_im_i^1m_i^{L_i-1}$ due to the definition of $A_i$. Then some non free point $a_j$ must 
by lying inside $a_im_i^1m_i^{L_i-1}$, for some $j < i$. If this is the case then $a_u$ must lie to the left of $\overrightarrow{m_j^{L_j-1}}a_j$,
which is impossible due to our induction hypothesis. It also follows that no two 
such triangles intersect. 
Since $\vert A \vert \leq l'$, the points of $A$ are depleted after less than $n$ iterations of the procedure. Now repeat an analogous procedure
for $B$ so that the free points of $B$ are on the right side of the triangles. After a total of $l'-1$ iterations, between $1$ and $l'-1$ 
rightmost points are left on $\lambda$. Since the free points of $B$ are also to the right of the other points of $B$, the convex hull of 
the free points of $B$ and the remaining points of $\lambda$ does not include any other point. 
So, either the convex hull of these points is a convex polygon, or all of these points are collinear (Figure \ref{linecyclfig}(a)).
In the first case we form the $l'^{th}$ cycle. In the second case, suppose that the rightmost cycle on $\lambda$ is 
$(b_{l'-2}, m_{l'-1}^1, \ldots, m_{l'-1}^{L_j-1},b_{l'-2})$. Let $b_p$ be the first free point of $B$ that the ray 
$\overrightarrow{m_{l'-1}^{L_j-1}b_{l'-2}}$ intersects first when rotated counterclockwise around $m_{l'-1}^{L_j-1}$.
We replace the cycle $(b_{L_j-1}, m_{l'-1}^1, \ldots, m_{l'-1}^{L_j-1},b_{l'-2})$ with $(b_{l'-2}, m_{l'-1}^1, \ldots, m_{l'-1}^{L_j-2}, b_p, b_{l'-2})$.
This releases the point $m_{l'-1}^{L_j-1}$, which cannot be collinear with the rest of the free points. So we now form the $l'^{th}$ cycle
(Figure \ref{linecyclfig}(b)).
 
 $\\ \\$
Now suppose that $\vert A \vert > l'-1$ and  $\vert B \vert > l'-1$. This actually means that $L_j>l'$. We implement the above procedure
on $A$, and it creates $l'-1$ cycles. The $l'^{th}$ cycle is to be constituted from all the points of $B$ and
the remaining points of $A$ and $\lambda$.
Observe that since all the free points of $A$ and $B$ are visible from all the remaining points of $l$, if $l$ has at least two remaining points
we can naturally construct the cycle. If $l$ has only one remaining point, then let $a_y$ be the first free point of $A$ encountered by rotating 
the ray $m_n^{L_j-1}a_n$ in an anticlockwise order around $m_n^{L_j-1}$. Replace the cycle 
$(a_{l'-1}, m_{l'-1}^1, m_{l'-1}^2, \ldots, m_{l'-1}^{L_j-2}, m_{l'-1}^{L_j-1}, a_{l'-1})$
by $(a_{l'-1}, m_{l'-1}^1, m_{l'-1}^2, \ldots, m_{l'-1}^{L_j-2}, a_y, a_{l'-1})$. Now the point $m_{l'-1}^{L_j-1}$ is a free point on $\lambda$, 
making the total number of free 
points on $\lambda$ at least $2$. We construct the $l'^{th}$ cycle as before.
\end{proof}
\begin{cor}
 If $P$ has $\sum ^l _{i=1} L_i$ points, not all $L_i =3$, 
 and there is some line $\lambda$ with exactly $\sum ^l _{i=1} L_i - l$ collinear points on it, then $P$ can be partitioned into disjoint cycles.
\end{cor}
\begin{lemma}
 The Necessary condition is not sufficient if  $\forall \ 1 \leq i \leq l \ L_i=3$. 
\end{lemma}
\begin{proof}
 Consider the point set in Figure \ref{sepcyclfig} (b). It can be seen that it satisfies the Necessary Condition, but removal of any three
 mutually visible points leaves out either three collinear points or three otherwise mutually invisible points.
\end{proof}
\begin{theorem}
  For any point set $P$ on $\sum_{i=1}^l  L_i$ points, $\exists L_j \geq 4$, the Necessary Condition is sufficient.
\end{theorem}
\begin{proof} 
Suppose $P$ has no more than $\sum_{i=1}^l L_i-l$ collinear points.
  Using the technique of Lemma \ref{disjcycl}, keep on separating out cycles in the ascending order of their length (so that one 
  of the remaining cycles always has length more than $3$)
  from $P$ till the condition is satisfied. 
  If the condition is satisfied throughout the process then we are done. Otherwise suppose that we are left with $P'$ such that 
  $\vert P' \vert = \sum_{i=1}^{l'-1} L_i$ and at least 
  $ \sum_{i=1}^{l'-1} L_i- l' +1$ points are collinear. Furthermore, assume that this is the first such instance, i.e.
  before separating out the last cycle $C_{l'}$, $P' \cup C_{l'}$ had at most 
  $\sum_{i=1}^{l'} L_i - l'$ collinear points. Then instead of separating $C_{l'}$ out,
  use the technique of Lemma \ref{bigline} to partition $P'$ into disjoint cycles.
\end{proof}
\subsection{2-factors of point visibility graphs}
Here we study the relationship between cycle partition of point sets and 2-factors of their visibility graphs.

\begin{lemma}
 Given a set of cycles $S = \{ C_1, C_2, \ldots, C_l \}$, not all $C_i$ of length $3$, a point visibility graph $G$ admits the corresponding
 2-factor if and only if $G$ does not have 
 any induced path on $(\sum ^ l _ {i =1} L_i) - l + 1$ vertices, where $L_i$ is the length of the cycle $C_i$. 
\end{lemma}
\begin{proof}
 If such a path exists, then one of the cyles must have all its vertices from the induced path, a contradiction. 
 If such a path does not exist, then no point set $P$ corresponding to $G$ can have $(\sum ^ l _ {i =1} L_i) - l + 1$
 collinear points. Thus, $P$ must have a disjoint cycle partition corresponding to $S$, which corresponds to a 2-factor of $G$.
\end{proof}
\begin{cor}
 If a point set $P$ has $l$ points on a line and at most $k$ points outside of it, where $l \geq 2k + 2$, 
 then every point set corresponding to the visibility 
 graph of $P$ has at least $l$ collinear points.
\end{cor}
\section{Clique partition} \label{secclfac}
The necessary condition for triangle partition is not sufficient for clique partitions where the cliques may have more than 
three points. Consider the graph drawn on a point set in Figure \ref{sggfig}.
Due to its structure, it is called a \emph{slanted grid graph} \cite{pvg-np-hard}. 
A slanted grid graph on $x$ points has a maximum independent size only of $O(\sqrt x)$.
The slanted grid graph in the figure has forty-four points with $k=4$ and $n=11$, 
and a maximum independent set of size five. However, it cannot be partitioned into copies of $K_4$, because $p_1$ and $p_2$ are not
contained in any $K_4$.
adjacent to triangles.
\begin{figure}[h] 
\begin{center}
\centerline{\hbox{\psfig{figure=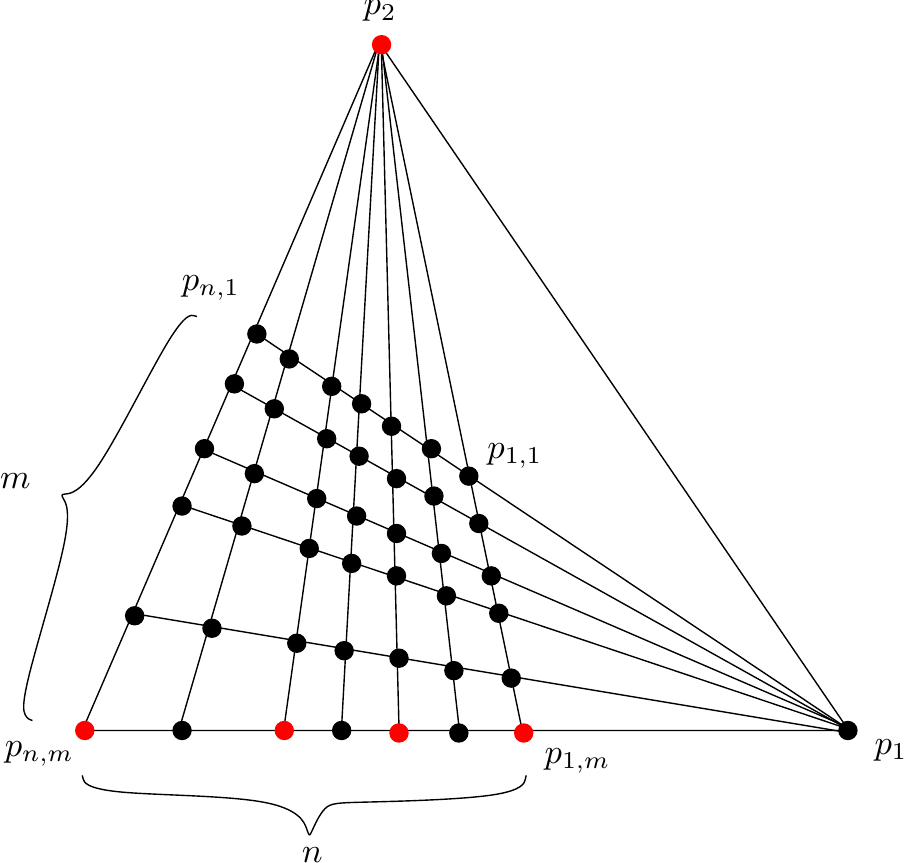,width=0.7\hsize}}}
\caption{A slanted grid graph on forty-four points. A maximum independent set is coloured in red}
\label{sggfig}
\end{center}
\end{figure}

$\\ \\$
In this section, we show that the problem of partitioning a given point set in the plane, into $k$-cliques for $k \geq 5$, is NP-hard. To show this,
we reduce \emph{3-occurrence SAT} to the problem. A 3-occurrence SAT formula is a SAT formula where each variable occurs at most $3$ times.
The 3-occurrence SAT problem is known to be NP-hard \cite{sat-tovey}.
\subsection{Construction of the reduction}
We provide a reduction of 3-occurrence SAT to partitioning a point set into copies of $K_5$.
We start with any given 3-occurrence SAT formula $\theta$, with variables $x_1, x_2, \ldots, x_n$ and clauses
$C_1, C_2, \ldots, C_m$. Wlog we assume that there is no variable in $\theta$ whose positive or negative literals solely 
constitute all of its occurrences. We also assume that each clause $C_i$, $2 \leq i \leq n-1$, has all variables different from those
in $C_{i-1}$ and $C_{i+1}$. This assumption is valid because any given 3-occurrence SAT formula can be transformed to such a formula
by adding a linear number of variables and clauses, with every variable and its negation occuring at least once each.
Let $n_1$ and $n_2$ be the number of variables that occur twice and thrice in $\theta$, respectively, so that $n_1 + n_2 = n$.

$\\ \\$
We now construct a point set $P$ from $\theta$ so that a partition of $P$ into 5-cliques is possible if and only if $\theta$ is satisfiable.
We do the following:
\begin{enumerate} [(a)]

\item
 Let $v =3n_1 + 4n_2 + n-1$, $b_n = 2(m-1)n_1 + (m-2)n_2 + 2(m-1)n_2 + 2(m-2)n_2 = 3mn + mn_2 - 3n - n_1$.
 Let  $e = bn + 2v − m − 1$, $b=e+m-1$ and $c = e + 2m - 2v$.
 
\item
Call the x-axis the \emph{clause line}.
Starting from the origin, from left to right, place $m$ points on the clause-line, unit distance apart. Each such point is called a \emph{clause-point}, the $k^{th}$ 
point representing $C_k$, and denoted as $cp_k$. 

\item
Consider the horizontal line with y-coordinate $-2$ and call it the \emph{extra-line}. 
Starting from the y-axis from left to right, we place $e$ points on the extra-line, each a unit distance apart.

\item
Consider the horizontal line with y-coordinate $-1$ and call it the \emph{blocking-line}. 
Starting from the y-axis from left to right, we place $b$ points on the blocking-line, each half a unit distance apart.

\item
Now consider the horizontal line with y-coordinate $-1.5$ and call it the \emph{variable-line}. 
Let $x'$ be the x-coordinate of the rightmost point on the extra-line. 
Starting from the point $(x'+1,-1.5)$ from left to right we place points on the variable-line, representing the variables of $\theta$.
\begin{enumerate}[(i)]
 \item 
 If $x_i$ and $\bar{x_i}$ occur in $C_j$ and $C_k$ respectively, then we place three points, $x_i^1$, $x_i^2 = \bar {x_i}^1$ and $\bar {x_i}^2$ to the right of
 all points placed so far on the variable-line, in the same consecutive order. We block the points $x_i^1$ and $\bar {x_i}^2$ from all points on the clause-line other than 
 $cp_j$ and $cp_k$ respectively. We block the point $x_i^2$ from all points on the clause-line other than 
 $cp_j$ and $cp_k$.
 
 \item 
 Wlog if there are two occurrences of $x_i$ and one occurrrence of $\bar {x_i}$, then we place four points, $x_i^1$, $x_i^2 = \bar {x_i}^1$, $x_i^3 = \bar {x_i}^2$ 
 and $x_i^4$ to the right of
 all points placeed so far on the variable-line, in the same consecutive order. Suppose that $x_i$ occurs in $C_j$ and $C_k$, and $\bar{x_i}$ occurs in $C_l$.
 We block the points $x_i^1$ and $x_i^4$ from all points on the clause-line other than 
 $cp_j$ and $cp_k$ respectively.
 We block the point $x_i^2$ from all points on the clause-line other than $cp_j$ and $cp_l$.
 We block the point $x_i^3$ from all points on the clause-line other than $cp_k$ and $cp_l$.
\end{enumerate}
\item
On the variable-line, introduce a new \emph{variable-blocker point} after all the points corresponding to a particular variable have been placed. Thus, 
there are $n-1$ variable-blocker points in total. Thus, now there are a total of $v$ points on the variable-line. 
Introduce blockers on the blocking-line so that no clause-point sees any of the variable-blocker points.

\item
Perturb the points on the variable-line slightly so that the corresponding blocking vertex blocks only a single pair of vertices. Thus there
are $b_n$
blockers in total used for the last two steps.

\item
Add $c$
more points to the clause-line, all to the right of the clause-points, such that they see all points not on the clause-line. 
\end{enumerate}
\subsection{Properties of the constructed point set}
We have the following lemmas based on the construction. 
\begin{lemma} \label{lemcompl}
 The construction of the point set can be completed in polynomial time.
\end{lemma}
\begin{proof} 
 The points placed on lattice points have integer coordinates. The length of these coordinates are only $O(log \ mn)$. The blockers
 are placed between the intersection of the blocking-line and the line passing through two such points. So, the coordinates of the 
 blockers are also of length $O(log \ n)$. After a blocker is placed on the blocking-line, it can coincide with some already placed
 blocker or block one or more pairs of points which are required to be visible. There are $O(mn)$ blockers in total. So
 there are only $O(mn)$ such undesirable positions. We first divide the variable-line into $4n$ intervals. Each of these intervals
 we further divide into $mn$ intervals. Clearly, the coordinates of the endpoints of the intervals are of length $O(log \ mn)$.
 Each of the perturbations can be achieved by assigning these coordinates to the variable-points. Hence the whole construction can
 be achived in $O(mn \ log \ mn)$ time.
\end{proof}
$\\ \\$
For the purpose of proving the later lemmas, we now study a related structure related to our construction.
Consider a partial grid $P_g$ on the lines $y=1$, $y=0$ and $y=-1$. We call these three lines the \emph{top}, \emph{middle} and \emph{bottom}
horizontal lines of $P_g$.
The partial grid starts from the y-axis and lies on the right side of it.
The points of the top and bottom horizontal lines of $P_g$ are only allowed to have nonnegative integer coordinates (Figure \ref{examplefig}(b)). 
For every point with coordinates $(x,y)$ in the top and bottom horizontal lines of $P_g$, where $x \neq 0$, there must also be 
a point with coordinates $(x-1,y)$ in $P_g$.
The points on the middle line of $P_g$ are only allowed to have coordinates of the form $(x,\frac{y}{2})$, where $x$ and $y$ are nonnegative integers. 
For every point on the middle line of $P_g$ with coordinates $(\frac{x}{2},0)$, where $x \neq 0$, there must also be a point with coordinates $(\frac{x-1}{2},0)$. 
Suppose $P_g$ has $p$ points on $y=1$ and $q$ points on $y=-1$. Then we have the following lemma.
\begin{lemma} \label{prtgrd}
 A total of $\lfloor \frac{p+q}{2} \rfloor$ points of $P_g$ on $y=0$ are necessary and sufficient to block all points of $P_g$ on $y=-1$ 
 from all points of $P_g$ on $y=-1$.
\end{lemma}
 \begin{proof}
  Consider points on the top and bottom lines on $P_g$ having coordinates $(x_1,1)$ and $(x_2,-1)$ respectively. The 
  point of intersection of the middle line and the line segment joining these two points is $(\frac{x_1+x_2}{2},0)$.
  Since $x_1 + x_2 \leq p + q$, this point is already in $P_g$. Now let $(x_3,0)$ be the coordinates of a point of $P_g$ on the middle line and wlog let $p \geq q$.
  This point blocks the points of $P_g$ with coordinates $(2x_3,1)$ and $(0,-1)$. 
 \end{proof}
 
 $\\ \\$
 Now we are ready to prove the following main lemmas on our construction.
\begin{lemma} \label{clvis}
  In our construction for the reduction, no clause-point can see any extra-point.
 \end{lemma}
\begin{proof}
 The clause-points, extra-points, and the blockers on lattice points of the blocking-line induce a partial grid. By Lemma \ref{prtgrd}, no
 clause-point can see any extra-point.
\end{proof}
 \begin{lemma} \label{oneside}
  $P_c$ can be $K_5$-partitioned if and only if $\theta$ has a satisfying assignment.
 \end{lemma}
  \begin{proof}
  Suppose $\xi$ is a $K_5$ partition of $P_c$. 
  By Lemma \ref{clvis}, no clause-point can see any extra-point. 
  So, a clause-point can see only its adjacent clause-points, all the blocking-points,
  and two variable-points for each of the variables that occur in its clause. 
  Suppose that a clause-point is a part of a $K_5$. 
  In the formula $\theta$, no consecutive clauses have the same variables.
  So, its two adjacent clause-points see neither each other,
  nor any of the variable points corresponding to the clause-point.
  Hence, the $K_5$ can contain only (a) the clause-point, (b) only two of the its variable-points corresponding to the
  same literal, since all the variable points are collinear and the points corresponding to each variable are separated by a 
  variable-blocker point.
  and (c) only two blocking-points, since all the blocking-points are collinear.
 
 $\\ \\$
  Given such a $K_5$-partition of $P_c$, the corresponding satisfying assignment of $\theta$ 
  can be constructed as follows. If some clause-point for $C_i$
  takes the variable-points for $x_j$ in its $K_5$, then assign $1$ to $x_j$. If $\bar{x_j}$ does not occur in $\theta$ then we are done for $C_i$.
  Otherwise, since $\theta$ is an instance of 3-occurrence SAT, $\bar{x_j}$ can occur in at most two clauses. But by construction of $P_c$,
  a variable-point from each of the pairs of variable-points representing the occurrences of $\bar{x_j}$, will coincide with the one 
  of the variable-points representing $x_j$ that the clause-point of $C_i$ took. So, no clause-point can include variable-points 
  representing $\bar{x_i}$
  in their $K_5$. An analogus reasoning holds if the clause-point for some $C_i$ takes the variable-point for $\bar{x_i}$. 
  Hence there is no conflict in assigning truth values to variables using the method described above.

$\\ \\$
Now we prove the other direction of the lemma.
 Consider a satisfying assignment of $\theta$. In $P_c$, start with the clause-point of $C_1$. The corresponding $K_5$ will contain the clause-point, the two leftmost
 points on the blocking-line, and the variable-points corresponding to a literal that is assigned $1$ in $C_1 \in \theta$. Similarly, for $C_i$, use 
 the $(2i-1)^{th}$ and $2i^{th}$ leftmost points on the blocking-line, and 
 the variable-points corresponding to a literal that is assigned $1$ in $C_i \in \theta$.
 Now each clause-point is in its respective $K_5$.
 
 $\\ \\$
 Now there are $v - 2m$ and $b+b_n - 2m$ remaining points on the variable-line and blocking-line respectively.
 Form a $K_5$ for each of the remaining points on the variable-line with two consecutive points on the blocking-line and extra-line, always
 choosing the leftmost points available.
 After this is done, $b+b_n - 2m -2(v - 2m)$ and 
 $e - 2m -2(v - 2m)$ points remain available on the blocking-line and extra-line respectively.
 Form a $K_5$ each for the remaining points on the extra-line, with two consecutive available points each from
 the blocking-line and clause-line.
 Due to the values of $c$, $e$, $b$, and $b_n$, all the points of $P_c$ are exhausted.
\end{proof}
 \begin{theorem} \label{genthm}
  For all $k\geq 5$, the $K_k$-partition problem for point sets on the plane is NP-hard.
 \end{theorem}
 \begin{proof}
 For $k=5$ we have Lemma \ref{oneside}. Suppose that $K$ is a greater odd number $5 + 2x$, then for a given instance of 3-occurrence SAT, first produce $P_c$ 
 for $K_5$-partition. Let $P_c$ have $5y$ points. Parallel to the clause-line and above it, draw $x$ lines of $2y$ points each, such that each
 new point is visible from every other new point and all points of $P_c$. 
 
 $\\ \\$
 Now consider the other case where $K$ is any greater even number $5 + 2x - 1$. First we discuss the case where $k=6$. 
 We assign new values to $b$, $c$ and $e$.
 Let $b= b_n$, $c=2b_n -v$ and $e = 2b_n - m$.
 Modify the construction for $k=5$ by intially placing $b$ points on the blocking-line, each a unit distance apart.
 Observe that, due to the above placement, a clause-point is visible from
 an extra-point if and only if the parity of their x-coordinates is different.
 So, since no clause-point sees to consecutive points on the extra-line, a clause-point can be placed into the same $K_6$
 with only one extra-point. 
  Also, this makes $b = \lfloor \frac{m+e}{2} \rfloor$.
 We ensure that the parity of $e$ and $m$ are the same, so the above relation becomes $b =  \frac{m+e}{2}$.
 As before, whenever possible, we choose the leftmost free points.
 Also, as before, this $K_6$ can contain only two blocking-points and variable-points each.
 After all the clause-points are exhausted thus, the number of available points 
 on the variable-line, blocking-line and extra-line, are $v - 2m$, $b+b_n-2m$ and $e-m$ respectively.
 After this, we place the remaining variable-points into copies of $K_6$, using two blocking-points, two extra-points and one new point from the 
 clause-line. So, now the number of available points 
 on the clause-line, blocking-line and extra-line, are $c-(v - 2m)$, $b+b_n-2m-2(v - 2m)$ and $e-m-2(v - 2m)$ respectively.
 We take two points each from these three lines and form copies of $K_6$.
 This is possible due to the new values of $b$, $c$ and $e$.

$\\ \\$
 If $k$ is any greater even number $5 + 2x - 1 = 6 + 2x - 2$, add $x-1$ new lines as in the case before.
\end{proof}

\section{Concluding Remarks} \label{secconcrem}
We have solved the problem of partitioning point sets into a set of polygons whose sizes are given.
and proved analogous results for their visibility graphs.
For clique partitions, when the sizes of given cliques are at least $5$, we have shown the problem to be NP-hard.
Our triangle-partition method indeed gives the solution 
a clique partition into triangles, but the result for $k \geq 4$ remains unknown.
The related problem of partitioning a point set into \emph{convex} polygons also remains unsolved.
\bibliographystyle{plain}
\bibliography{vis}
\end{document}